\documentclass[a4paper,10pt]{amsart}
\usepackage{amssymb}
\usepackage{amsthm}
\usepackage{graphicx}

\newtheorem{theorem}{Theorem}
\newtheorem{lemma}{Lemma}
\newtheorem{corollary}[theorem]{Corollary}

\def\R{\mathbb{R}}

\def\set#1{\left\{\, #1 \,\right\}}
\def\abs #1{\left| \,#1\, \right|}

\def\S{\mathbb{S}}

\def\calF{\mathcal{F}}

\def\calM{\mathcal{M}}

\begin{document}

\title [Minimal Moulton configurations]
{Generic uniqueness of the minimal\\
Moulton central configuration}

\author[R. Iturriaga]{Renato Iturriaga}
 \address{ CIMAT\\
  A.P. 402, 3600\\
  Guanajuato, Gto. M\'{e}xico.}
 \email{renato@cimat.mx}
 \thanks{Both authors thank ANR project WKBHJ
 for his support during the final part of the work.}
\author[E. Maderna]{Ezequiel Maderna}
 \address{Centro de Matem\'{e}tica\\
    Universidad de la Rep\'{u}blica, Uruguay.}
 \email{emaderna@cmat.edu.uy}

\subjclass[2010]{70F10}
\keywords{N-body problem, central configuration, genericity.}

\begin{abstract}
We prove that, for generic (open and dense) values of the masses,
the Newtonian potential function of the collinear N-body
problem has $N!/2$ critical values when restricted to a
fixed inertia level. In particular, we prove that for generic values
of the masses, there is only one global minimal Moulton configuration.
\end{abstract}

\maketitle

\begin{center}
\today
\end{center}

\section{Introduction}

In the N-body problem there is a family of solutions
that conserve the shape in the evolution in time.
Among these motions, those with zero angular momentum are
called homothetic motions.
They have the form \[x(t)=(r_1(t),\dots,r_N(t))=\phi(t)\, x_0\]
where $\phi (t)>0$ is a
solution of a one center problem in the line $\R_+$,
and $x_0$ a central configuration.
This kind of configurations can be defined in many equivalent ways,
say for instance as the critical points of the restrictions
of the potential function
\[U(x)=\sum_{i<j}\frac{m_i\,m_j}{r_{ij}}\,,\]
to the level sets of the moment of inertia
\[I(x)=\sum_{i=1}^N m_i\,r_i^2\;.\]
It turns out that these configurations have \textsl{center of mass at the origin}.
For some authors there is an extended notion of central configuration with
respect to its center of mass. These are only translation of the first ones.

In this paper we will be interested in the collinear
N-body problem, therefore a configuration
$x=(r_1,r_2,\dots,r_N)\in\R^N$ will represent the vector
of positions of the bodies, which are supposed to be point particles,
each with mass $m_i>0$, and contained in a straight line.
As usual, $r_{ij}=\abs{r_i-r_j}$ will denote the distance
between the bodies $r_i$ and $r_j$.

When the bodies evolve in a space of dimension $k>1$
not much is know about the geometry of central configurations.
Not even know in general if there exist only a finite number
-- modulo similitude  -- of central configurations.
One of the most recent works on this topic, due to
Albouy and Kaloshin \cite{AK}, shows the generic
finiteness in the case of five bodies in the plane,
that is, excluding the situation in which the vector
of masses $m=(m_1,\dots,m_5)$
belongs to a given subvariety of $\R_+^5$.

In contrast, for dimension $k=1$,  the problem is solved. The first step was given
by Euler who solved the case of three bodies see \cite{Eu}.
Moulton solved the problem for arbitrary number of masses.
More precisely, he proved in \cite{Mou} that if we identify
configurations which are homothetic by a positive factor,
then there are exactly $N!$ equivalence classes of critical
points, each one corresponding to an order $\sigma\in S_N$
 of the bodies in the line.
As we will explain, they are all nondegenerate local minima.
See also the appendix on Moulton's theorem in the paper by
Smale \cite{S}.

The mass vector $m=(m_1,\dots,m_N)\in\R_+^N$
is a parameter which determines the potential function $U$ and
the moment of inertia $I$. Thus the mass vector also determines
the central configurations. Before stating our result, let us
recall some well known equivalent definitions of central configuration.
Once the mass vector is fixed, we say that a configuration
$x_0\in\R^N$  without collisions (that is, such that
$U(x_0)<+\infty$) is a \textsl{central configuration} if and only if
one of the following equivalent conditions is satisfied:
\begin{itemize}
 \item [(a)] $x_0$ is a critical point of $U_0$, the restriction of
 $U$ to the level set of $I$ which contains $x_0$.
 \item [(b)] $x_0$ is a critical point of the homogeneous function
 (of zero degree) \[\tilde U=U\,I^{1/2}\,.\]
 \item [(c)] $x_0$ is a critical point of the function $U+\lambda I$
 for some value of $\lambda>0$.
\end{itemize}
Note that if $x_0$ is a central configuration then $r\,x_0$ is also
a central configuration for every $r\neq 0$.
Moreover, the notion of \textsl{nondegenerate} central configuration
refers to the first condition. More precisely, if $I(x_0)=k$ then $x_0$
is a nondegenerate central configuration when $x_0$ is a nondegenerate
critical point of the restriction of $U$ to the ellipsoid
$S_k=\set{x\mid I(x)=k}$.

The main result of the present note is the following theorem
and his corollary.

\begin{theorem}\label{minMoulton}
There is an open and dense set of mass vectors
$A\subset\R_+^N$ such that, if $m\in A$
then the function $\tilde U$ has $N!/2$ critical values.
\end{theorem}

\begin{corollary}
There is an open and dense set of mass vectors for which
the collinear $N$-body problem has only one global minimal
configuration.
\end{corollary}

Of course, the uniqueness in the statement of the corollary
refers to the similarity classes of central configurations,
that is to say, once we identify configurations
which are homothetic by a non zero factor.
Thus there are $N!/2$ different central configurations in this sense,
and generically only one of them is minimal.

In contrast, the number of critical values can be less than $N!/2$
for some values of the mass vector.
It is clear that if two masses are equal, and $N>3$,
then commutation of the corresponding bodies
gives an extra symmetry of the problem which is not
induced by an spacial isometry. In that case
it is also clear that we must have at least two
non similar minimal configurations, and at most
$N!/4$ critical values of the potential function
restricted to any inertia level.
If all the masses
are equal, the action of the full symmetric group
preserves the set of central configurations, which in turn
implies that the restriction of the potential function
to any inertia level has only one critical value, that is,
the potential takes the same value at every normalized
central configuration.

Before beginning the proof of the theorem, let us explain
our special interest in minimal central configurations.
They appear repeatedly in the recent literature on the
general $N$-body problem. More precisely, the minimality
condition is often necessary to apply global variational methods.
Indeed, in \cite{MV} the second author and Venturelli
have proved that if $\alpha$ is a given minimal
configuration normalized in the sense that $I(\alpha)=1$,
then for any configuration $x_0$ there is at least one
motion $x(t)$ starting from $x_0$ which is completely parabolic
for $t\to +\infty$, and whose normalized configuration
$x(t)\,I(x(t))^{-1/2}$ converges to $\alpha$. More recently,
Percino and S\'{a}nchez-Morgado \cite{PSM} built the Busemann functions
associated to each minimal configuration. This last result improves the
previous one, because each of these functions is provided with a lamination
of completely parabolic motions which are asymptotic to the minimal configuration.

In higher dimensions, as we already said, very little is known about the number of minimal central
configurations modulo similitude. However, at the risk of being bold and naive, it seems  natural to expect that generically
in the masses there should be only one minimal configuration. This is true for instance
when the dimension of the Euclidean space in which
the bodies move is $k\geq 2$ and the number of bodies $N$ does not exceeds
$k+1$. In this case we have, \textsl{for any choice of the masses}, only one minimal configuration
in which all the mutual distances $r_{ij}$ are equal. The main result in this work
shows that this is also true for arbitrary number of bodies and generic masses
in the collinear case. 

The proof of theorem \ref{minMoulton} is divided in several lemmas
which shall be established in the next section. The first two are given
for the sake of completeness even if they are well known. More precisely, 
these two lemmas contain a proof of Moulton's theorem which includes the
analytic dependence on the mass vector.

\section{Proof}

We begin by recalling a very useful and well known
way to normalize central configurations which was proposed
by Yoccoz at a conference in Palaiseau (\cite{Yoccoz}).
It is clear that $z\in\Omega$ is a central configuration if and only if
there exits $\lambda\in\R$ such that
\[\nabla U(z)+\lambda\,\nabla I(z)=0\,.\]
Since the functions $U$ and $I$ are homogeneous of degree
$-1$ and $2$ respectively, we deduce that
\begin{eqnarray*}
0&=&\left<\nabla U(z),z\right>+\lambda\,\left<\nabla I(z),z\right>\\
&=&-U(z)+2\lambda\,I(z)
\end{eqnarray*}
hence that $\lambda=U(z)/2I(z)$. We also see that $z$
is a central configuration if and only if $\mu z$
is a central configuration for all $\mu>0$, and that $\lambda(\mu z)=\mu^{-3}\lambda(z)$.
Therefore we conclude that there are two natural ways to normalize
the size of a central configuration: fixing the value of the moment of inertia, or
fixing the value of $\lambda$. The advantage of the second one is that
the normalized configuration is a critical point of the function
$U+\lambda\,I$ in the open set $\Omega$ rather than a critical point
of the restriction of $U$ to some level set of the moment of inertia.

Our first lemma proves the uniqueness and the analytical dependence
on the masses, of the normal central configuration once fixed the
ordering of the bodies.
Let us introduce before some convenient notation.

First, since we will consider varying masses, it will be convenient
to use the notation $U_m(x)=U(x,m)$ and $I_m(x)=I(x,m)$ for the values
at $x$ of the potential function and the moment of inertia respect
to the origin respectively.
Note that both functions $U$ and $I$ are real analytic
functions in $\Omega\times\R_+^N$.

Finally, as usual, $S_N$ will denote the group of bijections
of the set $\set{1,\dots,N}$ into itself. Each element of $S_N$
is therefore identified with an ordering of the $N$
bodies in the oriented straight line.
For $\sigma\in S_N$ we define the open set $\Omega_\sigma$ as the set
of configurations of $N$ bodies in the oriented line with the ordering
prescribed by $\sigma$, that is to say,
\[\Omega_\sigma=
\set{x=(r_1,\dots,r_n)\mid r_{\sigma(1)}<\dots<r_{\sigma(N)}}\,.\]
In other words, $\sigma(i)=j$ means that the mass $j$ occupies the place $i$
from left to right.
It is clear that the set $\Omega\subset\R^N$ of configurations without
collisions is the disjoint union of the above sets. Thus $\Omega$
has $N!$ connected components.

\begin{lemma}[Moulton's theorem] \label{Moulton-analitico}
For each $\sigma\in S_N$ there is a real analytic function
\[x_\sigma:\R_+^N\to \Omega_\sigma\]
such that $x_\sigma(m)$ is the unique central configuration
in $\Omega_\sigma$ for the collinear $N$-body problem with
mass vector $m$ such that $I_m(x_\sigma(m))=1$.
\end{lemma}

\begin{proof}
We will prove that for any value of $m\in\R_+^N$ and any
$\sigma\in S_N$ the function \[W_m=U_m+I_m\] has a unique
critical point in $\Omega_\sigma$. Clearly $W_m$ is a proper function
over each convex set $\Omega_\sigma$. Indeed, for each $K>0$,
$U_m(x)\leq K$ implies that $x$ is in the closed set
\[\set{(r_1,\dots,r_N)\in\R^N\mid
K\,r_{ij}\geq m_i\,m_j >0\textrm{ for all } 1\leq i<j \leq N}\subset\Omega\,,\]
and $\set{x\mid I_m(x)\leq K}$ is a compact subset of $\R^N$.
On the other hand $W_m$ is strictly convex in $\Omega$.
A simple computation shows that
\[\frac{\partial^2W_m}{\partial r_i^2}(x)=
2m_i+\sum_{k\neq i}2\,m_i\,m_k\,r_{ik}^{-3}\;\;
\textrm{ and that }
\;\;\frac{\partial^2W_m}{\partial r_i\,\partial r_j}(x)=
-2\,m_i\,m_j\,r_{ij}^{-3}\] when $i\neq j$.
Thus given $x=(r_1,\dots,r_N)$ and $y=(s_1,\dots,s_N)$ we can write
\[\left<\,y,D^2W_m(x)\,y\,\right>=
2\,\sum_{i<j}m_i\,m_j\,r_{ij}^{-3}\,(s_i-s_j)^2\;+2\,I_m(y)\,,\]
which implies that the spectrum
of the Hessian matrix is
uniformly bounded from below by $2\,m_0$ where
$m_0=\min\set{m_1\dots,m_N}>0$.
The same conclusion can be obtained by application of the
Gershgorin circle theorem (see \cite{FeVa}).
Therefore we deduce that
the function $W_m$ has one and only one critical point
at each component $\Omega_\sigma$ of $\Omega$.
We will call $c_\sigma(m)$ this critical point.
We have that $c_\sigma(m)$ is the unique central configuration
in $\Omega_\sigma$ such that $\lambda_m(c_\sigma(m))=1$.

The map $c_\sigma:\R_+^N\to\Omega_\sigma$ is real analytic
because it is also defined by the real analytic implicit function theorem
(see for instance chapter 6 in \cite{KP}),
applied to the real analytic function
\[F_\sigma:\Omega_\sigma\times\R_+^N\to\R^N\]
given by
\[F_\sigma(x,m)=
\frac{\partial U}{\partial x}(x,m)+\frac{\partial I}{\partial x}(x,m)=
\nabla W_m(x)\,.\]
We know that the necessary condition to apply the implicit
function theorem is satisfied since
\[\frac{\partial F_\sigma}{\partial x}(x,m)=
D^2W_m(x)\]
is the Hessian matrix of the function $W_m$
and we already know that is positive definite at every point.

In order to finish the proof, we write as a function of $m$
the corresponding central configuration with unitary moment
of inertia. Indeed, since $\lambda_m(c_\sigma(m))=1$ we have
that
\[2\,I(c_\sigma(m),m)=U(c_\sigma(m),m)\,,\]
and therefore
\[x_\sigma(m)=\sqrt{2}\,c_\sigma(m)\,U(c_\sigma(m),m)^{-1/2}\]
defines a real analytic function which gives, for each
value of the mass vector $m$ the unique central configuration
in $\Omega_\sigma$ with moment of inertia equal to $1$.
\end{proof}

Now we will prove that the collinear central configurations, also called
Moulton configurations, are local minima of
$\tilde U_m=U_m\,I_m^{1/2}$. Note that $\tilde U_m(x)$ is the
value of the potential $U_m$ at the normalized configuration
$I_m(x)^{-1/2}x$. Moreover, if we call
\[\S_m=\set{x\in\R^N\mid I_m(x)=1}\]
then every central configurations in $\S_m$ is a nondegenerate
local minimum of the restriction $U_m\mid_{\S_m}$, and a
global minimum on each component $\Omega_\sigma\cap\S_m$. We give
the proof of this well known fact for the sake of completeness.
We will use the arguments in the proof of the previous lemma.

\begin{lemma}\label{nondegenerate-local-minima}
Given $m\in\R_+^N$ and $\sigma\in S_N$ let us
write $\Sigma=\Omega_\sigma\cap\S_m$ for the set of normal configurations
with order $\sigma$. The function $U_m\mid_{\,\Sigma}$ has a unique
global minimum which is nondegenerate.
\end{lemma}

\begin{proof}
We already know that $U_m\mid_\Sigma$ has a unique critical point, thus
we only have to prove that it is a nondegenerate minimum. The critical point
is the point $x_\sigma$ in the previous lemma, so we have
\[x_\sigma=\sqrt{2}\,c_\sigma\,U_m(c_\sigma)^{-1/2}\]
where $c_\sigma$ is the unique critical point of $W_m=U_m+I_m$ in $\Omega_\sigma$.
Now we consider the map $\varphi:\Sigma\to\Omega_\sigma$ given by
\[\varphi(x)=\left(\frac{U_m(x)}{2}\right)^{1/3}x\,.\]
Clearly, $\varphi$ is a smooth embedding
which satisfies $\varphi(x_\sigma)=c_\sigma$, as shown in figure 1.

\vspace{.5cm}
\begin{figure}[h]
\centering
\includegraphics{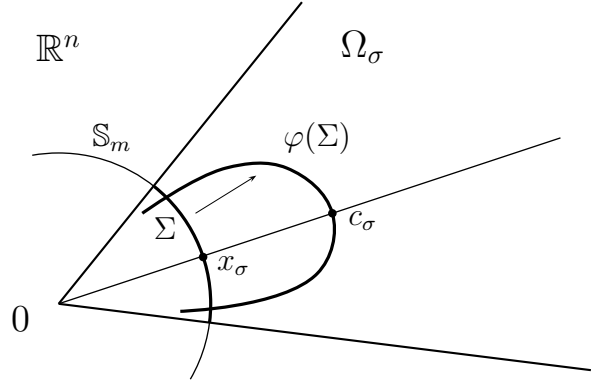}
\caption{The two different normalizations of a central configuartion.}
\end{figure}

\noindent
Moreover, for any $x\in\Sigma$ we have
\[U_m(\varphi(x))=2^{1/3}U_m(x)^{2/3}\;,\;\;\;
I_m(\varphi(x))=(1/4)U_m(x)^{2/3}\]
hence
\[U_m(x)=a\,W_m(\varphi(x))^{3/2}\]
for some constant $a>0$. This proves that $x_\sigma$ is a
nondegenerate minimum of $U_m\mid_\Sigma$ because $W_m$ has
a nondegenerate minimum at $c_\sigma=\varphi(x_\sigma)$ and
$W_m(c_\sigma)>0$.
\end{proof}

From now on, we will denote $\calM_N(\sigma,m)$ the
minimal value of the potential function $U_m$ restricted
to $\Sigma=\Omega_\sigma\cap\S_m$, the set of normal configurations
of $N$ bodies in the oriented line with a given order
prescribed by a permutation $\sigma\in S_N$. Thus we have
$\calM_N(\sigma,m)=U_m(x_\sigma)$.

We will say that $\sigma\in S_{N+k}$ is \emph{compatible} with
$\sigma_0\in S_N$ whenever for every
\[x=(r_1,\dots,r_N,r_{N+1},\dots,r_{N+k})
\in \Omega_\sigma\subset\R^{N+k}\]
we have
\[y=(r_1,\dots,r_N)\in\Omega_{\sigma_0}\subset\R^N\,.\]
Of course the condition can be written in terms of
$\sigma$ and $\sigma_0$ exclusively. More precisely,
taking into account that the value $\sigma(n)$ is the
number of the body in the $n$-th place from, it is easy to see
that $\sigma$ is compatible with $\sigma_0$ if and only if
the function
\[\sigma\mid_{\set{1,\dots,N}}\circ\;\sigma_0^{-1}:\set{1,\dots,N}\to
\set{1,\dots,N,N+1,\dots,N+k}\] is increasing.

\begin{lemma} \label{aproximacion}
Assume that $m_0=(m_1,\dots,m_N)\in\R_+^N$ and $\sigma\in S_N$
are given and that $\tau\in\S_{N+K}$ is compatible with $\sigma$.
If for $\epsilon>0$ we define the mass vector
\[m(\epsilon)=(m_1,\dots,m_N,
\epsilon\,m_{N+1},\dots,\epsilon\,m_{N+K})\in\R_+^{N+K}\]
then we have
\[\lim_{\epsilon\to 0}\calM_{N+K}(\tau,m(\epsilon))=
\calM_N(\sigma,m_0)\,.\]
\end{lemma}

\begin{proof}
Let $(\epsilon_n)_{n>0}$ be a minimizing sequence for
$\calM_{N+K}(\tau,m(\epsilon))$.
This means that $\epsilon_n\to 0$ and that
\[\liminf_{\epsilon\to 0}\calM_{N+K}(\tau,m(\epsilon))
=\lim_{n\to \infty}\calM_{N+K}(\tau,m(\epsilon_n))\,.\]
Now, for each $n>0$, we define $x_n\in\R^{N+K}$ as the unique normalized
central configuration of the $N+K$ bodies given by lemma \ref{Moulton-analitico},
for the mass vector $m(\epsilon_n)$ and the ordering given by $\tau$.
Thus, for each $n>0$ we have
\[I(x_n,m(\epsilon_n))=1\;\;\textrm{ and }\;\;
\calM_{N+K}(\tau,m(\epsilon_n))=U(x_n,m(\epsilon_n))\,,\]
where the last equality is due to lemma \ref{nondegenerate-local-minima}.
Moreover, if we write
\[x_n=(r_1^n,\dots,r_N^n,r_{N+1}^n,\dots,r_{N+K}^n)\,,
\;\;\textrm{ and }\;\;y_n=(r_1^n,\dots,r_N^n)\,,\]
then the compatibility of $\tau$ with $\sigma$ says that
the configuration $y_n$ has the ordering given
by the permutation $\sigma$.
The configuration $y_n$ is not normalized for the vector mass
$m_0$ as it is clear that $I(y_n,m_0)<1$. However, if we define
\[\alpha_n=m_{N+1}\,(r_{N+1}^n)^2+\dots+
m_{N+K}\,(r_{N+K}^n)^2\,,\]
we can write
\[I(y_n,m_0)=I(x_n,m(\epsilon_n))-\epsilon_n\,\alpha_n=
1-\epsilon_n\,\alpha_n\,\]
so the normalization of $y_n$ gives the configuration
$z_n= (1-\epsilon_n\,\alpha_n)^{-1/2}\,y_n$,
and we have
\[U(z_n,m_0)=(1-\epsilon_n\,\alpha_n)^{1/2}\;U(y_n,m_0)\,.\]
On the other hand, we have that
\[U(x_n,m(\epsilon_n))=U(y_n,m_0)+
\sum_{i=1}^N\sum_{j=N+1}^{N+K}
\frac{\epsilon_n\,m_i\,m_j}{\abs{r_i^n-r_j^n}}+
\sum_{N+1\leq i<j\leq N+K}
\frac{\epsilon_n^2\,m_i\,m_j}{\abs{r_i^n-r_j^n}}\,.\]
Since $(1-\epsilon_n\,\alpha_n)^{1/2}<1$, we deduce that
\[U(z_n,m_0)<U(y_n,m_0)<U(x_n,m(\epsilon_n))\,.\]
Thus, given that $\calM_N(\sigma,m_0)\leq U(z_n,m_0)$, we conclude that
\[\calM_N(\sigma,m_0)<U(x_n,m(\epsilon_n))
=\calM_{N+K}(\tau,m(\epsilon_n))\,.\]
Taking the limit for $n\to\infty$ we obtain the inequality
\[\calM_N(\sigma,m_0)\leq
\liminf_{\epsilon\to 0}\calM_{N+K}(\tau,m(\epsilon))\,.\]

We fix now $\delta>0$ and we define $z=(r_1,\dots,r_N)$ as
the unique normal central configuration for the mass vector $m_0$
and ordering prescribed by $\sigma$. In particular we have
$\calM_N(\sigma,m_0)=U(z,m_0)$ by lemma \ref{nondegenerate-local-minima}.
We will prove that the inequality
\[\calM_{N+K}(\tau,m(\epsilon))\leq U(z,m_0)+\delta\]
is satisfied whenever $\epsilon>0$ is small enough.
This will finish the proof, since it implies that
\[\limsup_{\epsilon\to 0}\calM_{N+K}(\tau,m(\epsilon))\leq
\calM_N(\sigma,m_0)\,.\]

Since $\tau$ is compatible with $\sigma$, we can add to the
configuration $z=(r_1,\dots,r_N)$ the positions of $K$ bodies,
in such a way that the ordering of the resulting extended configuration
$y=(r_1,\dots,r_N,r_{N+1},\dots,r_{N+K})$ is given by $\tau$.
We shall call $r_0$ the minimal distance between the positions
in the configuration $y$, that is to say,
\[r_0=\min\set{\abs{r_i-r_j}\mid 1\leq i<j\leq N+K}>0\,.\]
We will also consider the moment of inertia of the configuration
$y$ with respect to the mass vector $m(\epsilon)$, and we will denote
it by $I_\epsilon$. Thus we can write
\[I_\epsilon=I(y,m(\epsilon))=I(z,m_0)\,+\epsilon\,\alpha\,,\]
where
\[\alpha=m_{N+1}\,r_{N+1}^2+\dots+m_{N+K}\,r_{N+K}^2\,.\]
Moreover, since $z$ is a normal configuration for $m_0$, we can
write $I_\epsilon=1+\epsilon\,\alpha$.
Thus, normalizing
$y$ with respect to the mass vector $m(\epsilon)$
we obtain the configuration
\[x_\epsilon=I_\epsilon^{-1/2}y\,.\]

We observe now that the homogeneity gives
$U(x_\epsilon,m(\epsilon_n))=I_\epsilon^{1/2}\,U(y,m(\epsilon))$
and that
\[U(y,m(\epsilon_n))=
U(z,m_0)+
\sum_{i=1}^N\sum_{j=N+1}^{N+K}
\frac{\epsilon\,m_im_j}{\abs{r_i-r_j}}+
\sum_{N+1\leq i<j\leq N+K}
\frac{\epsilon^2\,m_im_j}{\abs{r_i-r_j}}\,.\]
Hence we deduce the upper bound
\[U(x_\epsilon,m(\epsilon))\leq
I_\epsilon^{1/2}\,
\left(U(z,m_0)+
N\,K\,\frac{\epsilon\,\mu^2}{r_0}+
\frac{K(K-1)}{2}\,\frac{\epsilon^2\,\mu^2}{r_0}\right)\,,\]
where $\mu=\max\set{m_1,\dots,m_{N+K}}$.
Therefore, since the right hand of the previous inequality
is a continuous function of $\epsilon$, and
$\calM_{N+K}(\tau,m(\epsilon))\leq U(x_\epsilon,m(\epsilon))$,
we conclude that there is $\epsilon_0>0$ such that
\[\calM_{N+K}(\tau,m(\epsilon))<U(z,m_0)+\delta\]
whenever $\epsilon<\epsilon_0$, as we wanted to prove.
\end{proof}

\begin{lemma} \label{L3BP}
There is $\mu>0$ for which
\[\calM_3(id,(1,\mu,1))\neq
\calM_3((2,3),(1,\mu,1))=\calM_3(id,(1,1,\mu))\,.\]
\end{lemma}

\begin{proof}
Let us first compute $\calM_3(id,(1,\mu,1))$.
The symmetry of the mass vector
implies that the central configurations in $\Omega_{id}$
are also symmetric. This means that
the configurations have the form $x_r=(-r,0,r)$
with $r>0$. Computing the potential function
and the moment of inertia we get
\[I(x_r)=2\,r^2\;\;\textrm{ and }\;\;
U(x_r)=\frac{1}{2r}+\frac{2\,\mu}{r}\,.\]
So the normal central configuration for this order
of the masses is $(-1/\sqrt{2},0,1/\sqrt{2})$.
We deduce that
\[\calM_3(id,(1,\mu,1))=\frac{\sqrt{2}}{2}+2\sqrt{2}\,\mu\,.\]

The second distribution of masses is not symmetric.
However, Euler has showed (see \cite{Eu}, or \cite{AF}
for a modern reference) that
up to a translation and rescale, a
central configuration for the mass vector $(m_1,m_2,m_3)$ and
order $\sigma=id$ is $(0,1,1+s)$, where $s$ is the unique
positive root of the polynomial
\begin{eqnarray*}
p(s)& = &-(m_1+m_2)s^5-(3\,m_1+2\,m_2)s^4-(3\,m_1+m_2)s^3+\\
&&\\
&&+(m_2+3\,m_3)s^2+(2\,m_2+3\,m_3)s+(m_2+m_3)\,.
\end{eqnarray*}
Since in our case we have $m_1=m_2=1$ and $m_3=\mu$ the polynomial
becomes
\[p(s)=-2\,s^5-5\,s^4-4\,s^3+
(1+3\,\mu)s^2+(2+3\,\mu)s+(1+\mu)\,.\]
We claim that there is $\mu>0$ for which $(0,1,3)$ is a
translated central configuration. Therefore $s=2$ must
be a root of this polynomial,
which gives rise to the linear equation
\[p(2)=19\mu-171=0\] whose solution is $\mu=9$.
We conclude that, the central
configurations, for the mass vector $(1,1,9)$ and
the ordering given by $\sigma=id$, have the form $y_r=(0,r,3r)$
with $r>0$. Using the Leibnitz
formula for the moment of inertia with respect to
the center of mass we avoid to translate the configuration.
More precisely, we have
\begin{eqnarray*}
I_G(y_r)&=&\frac{1}{m_1+m_2+m_3}\,
\left(m_1\,m_2\,r_{12}^2+m_1\,m_3\,r_{13}^2+
m_2\,m_3\,r_{23}^2\right)\\
&&\\
&=&\frac{1}{11}\left(r^2+9(3r)^2+9(2r)^2\right)
=\frac{118}{11}\,r^2\,.
\end{eqnarray*}
In particular, the central configuration with moment
of inertia $I_G=1$ is, up to a translation, the
configuration $y_r$ for $r=(11/118)^{1/2}$. Now we
can compute the value of the potential function in
this configuration, and we get
\[U(y_r)=\frac{1}{r}+\frac{9}{2r}+\frac{9}{3r}=
\frac{51}{6}\left(\frac{118}{11}\right)^{1/2}\,.\]

Therefore the lemma is proved, since for $\mu=9$ we
have computed
\[\calM_3(id,(1,9,1))=\frac{37}{\sqrt{2}}\]
and
\[\calM_3((2,3),(1,9,1))=\calM_3(id,(1,1,9))=
\frac{51}{6}\left(\frac{118}{11}\right)^{1/2}\,.\]
\end{proof}

The last lemma we will need in the proof of the theorem
is purely combinatorial and characterizes the fact that two
permutations are not equal nor symmetric.
Let us introduce first simplifying notations. 
If $\sigma\in S_N$ is a given permutation, then
we will write $\bar\sigma$ to denote the permutation
corresponding to the inverse order. More precisely,
$\bar\sigma$ is defined by $\bar\sigma(k)=\sigma(N+1-k)$.
Moreover, given $\sigma\in S_N$ and numbers
$i,j,k\in\set{1,\dots,N}$, we will say that $\sigma(i)$
is between $\sigma(j)$ and $\sigma(k)$ if either
$\sigma(j)<\sigma(i)<\sigma(k)$ or
$\sigma(k)<\sigma(i)<\sigma(j)$.

\begin{lemma}\label{combinatorio}
If $\sigma $ and $\tau$ are two given permutations
then we have the following alternative:
either $\sigma=\tau$, $\sigma=\bar\tau$, or
there are three numbers $i,j,k\in\set{1,\dots,N}$
such that $\sigma(i)$ is between $\sigma(j)$
and $\sigma(k)$, but $\tau(i)$ is not between
$\tau(j)$ and $\tau(k)$.
\end{lemma}

\begin{proof}
Clearly, each one of the first two possibilities
in the triple alternative excludes the others.
Thus it suffices to show that if the third
possibility is not satisfied then one of the
two first must be true.

It is not difficult to see that if the third
possibility is not satisfied then $\sigma\circ\tau^{-1}$
is a monotone bijection. On the other hand, the only permutations on the
$n$ numbers $\set{1,\dots,n}$ which are monotone are the identity and the
inversion. Thus, we must have $\sigma=\tau$ or $\sigma=\bar\tau$.
\end{proof}

\begin{proof}[\textbf{Proof of theorem \ref{minMoulton}}]
Recall that for each $\sigma\in S_N$, we denote $x_\sigma(m)$
the unique central configuration with ordering given by $\sigma$
and normalized in the sense that $I(x_\sigma(m),m))=1$. Therefore,
the set of critical values of the function $\tilde U$ is exactly
\[V_c(m)=\set{U(x_\sigma(m),m)\mid \sigma \in S_N}.\]
Thus we know that the number of critical values is a lower
semicontinuous function of $m\in\R_+^N$, so in particular
it is a continuous function over the set of maxima
\[A=\set{m\in\R_+^N\textrm{ such that } \abs{V_c(m)}=N!/2}\]
from which we conclude that this set is open.

In what follows we prove that $A$ is dense in $\R_+^N$.
Let us define, for each pair of permutations $\sigma,\tau\in S_N$,
the set
\[M_{\sigma,\tau}=
\set{m\in\R_+^N\mid U(x_\sigma(m),m)\neq U(x_\tau(m),m)}.\]
As a consequence of the
analyticity property proved in lemma \ref{Moulton-analitico}
we know that each one of these sets is either
open and dense, or empty.
We will prove that the empty case happens only if
$\sigma=\tau$ or $\sigma=\bar\tau$.
The proof of this claim finish the proof, since
\[A=\bigcap_{(\sigma,\tau)\in\calF} M_{\sigma,\tau}\]
where $\calF$ is the set of pairs $(\sigma,\tau)$
of non symmetric permutations, i.e. such that
$\sigma\neq\tau$ and $\sigma\neq\bar\tau$.
In order to prove the claim, we assume by contradiction
that $\sigma$ and $\tau$ are non symmetric permutations
and that the set $M_{\sigma,\tau}$ is however empty.
Thus we have
$U(x_\tau(m),m)=U(x_\sigma(m),m)$ for all $m\in\R_+^N$.

On the other hand, since $\sigma\neq\tau$ and
$\sigma\neq\bar\tau$ by lemma \ref{combinatorio}
(applied to the inverse permutations $\sigma^{-1}$ and $\tau^{-1}$)
we can assume without loss of generality that there are numbers
$1\leq i<j<k\leq N$ such that
\[\sigma^{-1}(i)<\sigma^{-1}(j)<\sigma^{-1}(k)\,,\]
and
\[\tau^{-1}(i)<\tau^{-1}(k)<\tau^{-1}(j)\,.\]
We can also assume, renumbering the bodies if necessary,
$i=1$, $j=2$ and $k=3$.
Now consider for small $\epsilon>0$ the mass vector
$m_\epsilon=(1,\mu,1,\epsilon,\dots,\epsilon)$ given by
where $\mu$ is the value of the mass given by lemma \ref{L3BP}.
By lemma \ref{nondegenerate-local-minima} we have
\[\calM_N(\sigma,m_\epsilon)=U(x_\sigma(m_\epsilon),m_\epsilon)
=U(x_\tau(m_\epsilon),m_\epsilon)=\calM_N(\tau,m_\epsilon)\]
for all $\epsilon>0$. Moreover applying lemma \ref{aproximacion}
with $N=3$, $\sigma_0=id$ and $\tau_0=(2,3)$ we have
\[\lim_{\epsilon\to 0}\calM_N(\sigma,m_\epsilon)=\calM_3(id,(1,\mu,1))\]
and
\[\lim_{\epsilon\to 0}\calM_N(\tau,m_\epsilon)=\calM_3((2,3),(1,\mu,1))\,.\]
This is impossible since it contradicts lemma \ref{L3BP}.
\end{proof}

\emph{Acknowledgements.}
The authors would like to thank the anonymous referees and Professor
Alain Albouy for their suggestions and comments. Following the suggestions,
we have included several improvements in the manuscript.

\end{document}